\documentclass[runningheads]{llncs}
\usepackage[utf8]{inputenc}
\usepackage[USenglish]{babel}
\usepackage{amsfonts,amsmath,mathtools,thm-restate}
\usepackage{tikz}
\usepackage{url}
\usetikzlibrary{automata,arrows,positioning}
\usepackage{graphicx}
\usepackage{url}
\usepackage{hyperref}
\usepackage{cleveref}
\usepackage{csquotes}
\usepackage{gastex}
\usepackage{xspace}
\usepackage{smallsec}

\newcommand{\init}{\iota}

\newcommand{\comp}{{\it comp}}
\newcommand{\zug}[1]{\langle #1 \rangle}

\newcommand{\A}{{\cal A}}

\newcommand{\B}{{\cal B}}
\newcommand{\T}{{\cal T}}

\newcommand{\RL}{{\cal R}_L}
\newcommand{\res}[1]{\langle #1  \rangle}

\newcommand{\DFA}{\mbox{\rm DFA}\xspace}

\newcommand{\DFAs}{\mbox{\rm DFAs}\xspace}

\DeclareMathOperator{\indx}{index}
\DeclareMathOperator{\DFAreal}{DFA}
\DeclareMathOperator{\violate}{Violate}
\DeclareMathOperator{\violatesep}{Violate}
\DeclareMathOperator{\NoDFAreal}{NoDFA}
\DeclareMathOperator{\DFAsep}{SepDFA}
\DeclareMathOperator{\NoDFAsep}{NoSepDFA}

\begin{document}

\title{Certifying DFA Bounds\\ for Recognition and Separation\thanks{This is the full version of an article with the same title that appears in the ATVA 2021 conference proceedings. The final authenticated publication is available online at \url{https://doi.org/[insert DOI]}. Orna Kupferman is supported in part by the Israel Science Foundation, grant No. 2357/19. Salomon Sickert is supported by the Deutsche Forschungsgemeinschaft (DFG) under project number 436811179.}}

\author{Orna Kupferman \and
Nir Lavee \and
Salomon Sickert}
\institute{School of Computer Science and Engineering, The Hebrew University, Israel.
\email{orna@cs.huji.ac.il}, \email{nir.lavee@mail.huji.ac.il}, \email{salomon.sickert@mail.huji.ac.il}}
\authorrunning{Orna Kupferman, Nir Lavee, and Salomon Sickert}
\maketitle              

\begin{abstract}
The automation of decision procedures makes certification essential. We suggest to use determinacy of turn-based two-player games with regular winning conditions in order to generate certificates for the number of states that  a deterministic finite automaton (DFA) needs in order to recognize a given language. Given a language $L$ and a bound $k$, recognizability of $L$ by a DFA with $k$ states is reduced to a game between Prover and Refuter. The interaction along the game then serves as a certificate. Certificates generated by Prover are minimal DFAs. Certificates generated by Refuter are faulty attempts to define the required DFA.
We compare the length of offline certificates, which are generated with no interaction between Prover and Refuter, and online certificates, which are based on such an interaction, and are thus shorter. We show that our approach is useful also for certification of separability of regular languages by a DFA of a given size. Unlike DFA minimization, which can be solved in polynomial time, separation is NP-complete, and thus the certification approach is essential. In addition, we prove NP-completeness of a strict version of separation.
\end{abstract}

\section{Introduction}

{\em Deterministic finite automata\/} (DFAs) are among the most studied computation models in theoretical computer science.
In addition to serving as an abstract mathematical concept, they are often the basis for specification and implementation of finite-state hardware and software designs \cite{Aut10}.
In particular, the theory of DFAs applies also to deterministic automata of infinite words that recognize {\em safety\/} languages, which are characterized by finite forbidden behaviors \cite{AS87,KV01d}.

A fundamental problem about DFAs is their {\em minimization}:
For $k \geq 1$, we say that a language $L \subseteq \Sigma^*$ is {\em $k$-DFA-recognizable\/} if there is a $k$-DFA, namely a DFA with at most $k$ states, that recognizes $L$. In the minimization problem, we are given a DFA $\A$ and a bound $k \geq 1$, and decide whether $L(\A)$, namely the language of $\A$, is $k$-DFA-recognizable.
DFAs enjoy a clean (and beautiful) theory of canonicity and minimization, based on a
{\em right-congruence} relation:  A language $L \subseteq \Sigma^*$ induces a relation $\sim_L \subseteq \Sigma^* \times \Sigma^*$, where for every two words $h_1,h_2 \in \Sigma^*$, we have that $h_1 \sim_L h_2$ iff for all words $t \in \Sigma^*$, we have that $h_1 \cdot t \in L$ iff $h_2 \cdot t \in L$.
By the Myhill-Nerode Theorem \cite{Myh57,Ner58}, the language $L$ is $k$-DFA-recognizable iff the number of equivalence classes of $\sim_L$ is at most $k$.
Moreover, a given DFA $\A$ can be minimized in polynomial time, by a fixed-point algorithm that merges states associated with the same equivalence class of $\sim_{L(\A)}$.

Another fundamental problem about DFAs is {\em separation}:  Given DFAs $\A_1$ and $\A_2$, and a bound $k \geq 1$, decide whether there is a $k$-DFA $\A$ that {\em separates} $\A_1$ and $\A_2$. That is, $L(\A_1) \subseteq L(\A)$ and $L(\A) \cap L(\A_2) = \emptyset$.  Finding a separator for $\A_1$ and $\A_2$ is closely related to the {\em DFA identification\/} problem. There, given sets $S_1,S_2 \subseteq \Sigma^*$ of positive and negative words, and a bound $k \geq 1$, we seek a $k$-DFA that accepts all words from $S_1$ and no word from $S_2$. DFA identification is NP-complete \cite{Gol78}, with numerous heuristics and applications \cite{TB73,HV10}.
NP-hardness of DFA separation can be obtained by a reduction from DFA identification, but for DFA separation with additional constraints, in particular strict separation, NP-hardness is open \cite{GGS17}.
Studies of separation include a search for regular separators of general languages \cite{CLMMKS18}, as well as separation of regular languages by weaker classes of languages, e.g., FO-definable languages \cite{PZ16} or piecewise testable languages \cite{CMM13}.

Let us return to the problem of DFA minimization, and assume we want to {\em certify} the minimality of a given DFA. That is, we are given a DFA $\A$ and a bound $k \geq 1$, and we seek a proof that $L(\A)$ is not $k$-DFA-recognizable.
The need to accompany results of decision procedures by a certificate is not new, and includes certification of a ``correct'' decision of a model checker   \cite{KV05d,ACOW18}, reachability certificates in complex multi-agent systems \cite{AL20}, and explainable reactive synthesis \cite{BFT20}.
Certifying that $L(\A)$ is not $k$-DFA-recognizable, we can point to $k+1$ words $h_1,\ldots, h_{k+1} \in \Sigma^*$ that belong to different equivalence classes of the relation $\sim_{L(\A)}$, along with an explanation why they indeed belong to different classes, namely words $t_{i,j} \in \Sigma^*$, for all $1 \leq i \neq j \leq k+1$, such that $h_i \cdot t_{i,j}$ and $h_j \cdot t_{i,j}$ do not agree on their membership in $L(\A)$.

The above certification process is {\em offline}: Refuter (that is, the entity proving that $L(\A)$ is not $k$-DFA-recognizable) generates and outputs the certificate without an interaction with Prover (that is, the entity claiming that $L(\A)$ is $k$-DFA-recognizable). In this work we describe an {\em interactive certification protocol}:\footnote{Note that while our certification protocol is interactive, the setting is different from that of an interactive proof system in computational complexity theory. In particular, our Prover and Refuter are both finite-state, they have complementary objectives, and no probability is involved.}
Given $\A$ and $k \geq 1$, Refuter and Prover interact, aiming to convince each other about the (non-)existence of a $k$-DFA for $L(\A)$.
Our approach offers two advantages over offline certification. First, the length of the certificate is shorter. Second, the interactive protocol can also be used for efficiently certifying bounds on the size of DFA separators. In addition, we solve the open problem of the complexity of deciding strict separation by a $k$-DFA. We show that it is NP-complete, and so are variants requiring only one side of the separation to be strict.

The underlying idea behind the interactive certification protocol is simple: Consider a language $L \subseteq \Sigma^*$ and a bound $k \geq 1$. We consider a {\em turn-based two-player game\/} between Refuter and Prover. In
each round in the game, Prover provides a letter from a set $[k]=\{1,2,\ldots,k\}$ that describes the state space of a DFA for $L$ that Prover claims to exist, and Refuter responds with a letter in $\Sigma \cup \{\#\}$, for a special reset letter $\# \not \in \Sigma$. Thus, during the interaction, Prover generates a word $y \in [k]^\omega$ and Refuter generates a word $x \in (\Sigma \cup \{\#\})^\omega$. The word $x$ describes an infinite sequence of words in $\Sigma^*$, separated by $\#$'s, and the word $y$ aims to describe runs of a $k$-DFA on the words in the sequence.
Prover wins if the described runs are legal: They all start with the same initial state and follow some transition function, and are consistent with $L$: There is a way to classify the states in $[k]$ to accepting and rejecting such that Prover responds with an accepting state whenever the word generated by Refuter since the last $\#$ is in $L$.  Clearly, if there is a $k$-DFA for $L$, then Prover can win by following its runs. Likewise, a winning strategy for Prover induces a $k$-DFA for $L$. The key idea  behind our contribution is that since the above described game is determined \cite{BL69}, Refuter has a winning strategy iff no $k$-DFA for $L$ exists. Moreover, since the game is regular, this winning strategy induces a finite-state transducer, which we term an {\em $(L,k)$-refuter}, and which generates interactive certificates for $L(\A)$ not being $k$-DFA-recognizable.

Consider a language $L$ with index $N$. Recall that the interaction between Refuter and Prover generates words $x \in (\Sigma \cup \{\#\})^\omega$ and $y \in [k]^\omega$. If $k < N$, Refuter can generate $x$ for which the responses of Prover in $y$ must contain a violation of legality or agreement with $L$. Once a violation is detected, the interaction terminates and it constitutes a certificate: an {\em informative bad prefix\/} \cite{KV01d} of the safety language of interactions in which Prover's responses are legal and agree with $L$. We show that the length of certificates generated by offline refuters is at most $O(k^2 \cdot N)$, whereas interaction reduces the length to $O(k^2 + N)$. We show that both bounds are tight. For separation, we describe a refuter that generates certificates of length at most $O(k^2 \cdot |\Sigma| + k \cdot (N_1 + N_2))$, where $N_1$ and $N_2$ are the indices of the separated languages.

Our interactive certification protocol has similarities with the interaction that takes place in {\em learning\/} of regular languages \cite{Ang87}, (see recent survey in \cite{Fis18}). There, a Learner is tasked to construct a DFA $\A$ for an unknown regular set $L$ by asking a Teacher queries of two types: Membership (``$w \in L$?'') and equivalence  (``$L(\A) = L$?''). In our setting, Refuter also wants to ``learn'' the $k$-DFA for $L$ that Prover claims to possess, but she needs to learn only a fraction of it from Prover -- a fraction that reveals that it does not actually recognize $L$. This is done with a single type of query (``what is the next state?''), which may give Refuter more information than the information gained in the learning setting.

\section{Preliminaries}

\subsection{Automata}
\label{prelim automata}

A {\em deterministic automaton on finite words\/} (\DFA, for short) is $\A=\zug{\Sigma,Q,q_0,\delta,F}$, where $Q$ is a finite set of states, $q_0 \in Q$ is an initial state, $\delta: Q\times \Sigma \rightarrow Q$ is a partial transition function, and $F \subseteq Q$ is a set of final states. We sometimes refer to $\delta$ as a relation $\Delta\subseteq Q\times \Sigma\times Q$, with $\zug{q,\sigma,q'}\in \Delta$ iff $\delta(q,\sigma)=q'$. A run of $\A$ on a word $w=w_1 \cdot w_2 \cdots w_m \in \Sigma^*$ is the sequence of states $q_0,q_1,\ldots,q_m$ such that $q_{i+1} = \delta(q_i,w_{i+1})$ for all $0 \leq i < m$. The run is accepting if $q_m \in F$. A word $w \in \Sigma^*$ is accepted by $\A$ if the run of $\A$ on $w$ is accepting. The language of $\A$, denoted $L(\A)$, is the set of words that $\A$ accepts. We define the {\em size\/} of $\A$, denoted $|\A|$, as the number of states that $\A$ has. For a language $L \subseteq \Sigma^*$, we use ${\it comp}(L)$ to denote the language complementing $L$, thus ${\it comp}(L)=\Sigma^* \setminus L$.

Consider a language $L \subseteq \Sigma^*$. For two finite words $h_1$ and $h_2$, we say that $h_1$ and $h_2$ are {\em right $L$-indistinguishable}, denoted $h_1 \sim_L h_2$, if for every $t \in \Sigma^*$, we have that $h_1 \cdot t \in L$ iff $h_2 \cdot t \in L$. Thus, $\sim_L$ is the Myhill-Nerode right congruence used for minimizing \DFAs. For $h \in \Sigma^*$, let $[h]$ denote the equivalence class of $h$ in $\sim_L$ and let $\res{L}$ denote the set of all equivalence classes.
When $L$ is regular, the set $\res{L}$ is finite and we use $\indx(L)$ to denote $|\res{L}|$. The set $\res{L}$ induces the \emph{residual automaton} of $L$, defined by $\RL=\zug{\Sigma,\res{L}, \Delta_L,[\epsilon],F}$, with $\zug{[h],a,[h \cdot a]} \in \Delta_L$ for all $[h] \in \res{L}$ and $a \in \Sigma$. Also, $F$ contains all classes $[h]$ with $h \in L$. The \DFA $\RL$ is well defined and is the unique minimal \DFA for $L$.

\begin{lemma}
\label{head bound}
Consider a regular language $L$ of index $N$.
	For every $1\le k\le N$, there is a set $H_{k}=\{h_1,\ldots,h_k\}$ of words $h_i \in \Sigma^*$ such that $h_i \not \sim_L h_j$ for all $1 \leq i \neq j \leq k$ and $|h_i|\le k-1$ for all $1 \leq i  \leq k$.
\end{lemma}

\begin{proof}
Let $H$ be a set of shortest representatives of the classes in $\res{L}$.
	If every word $h\in H$ has $|h|\le k-1$, we can define $H_{k}$ as
	an arbitrary subset of size $k$ of $H$. Otherwise, there
	exists $h\in H$ with $|h|\ge k$. Let $[h_{1}],\ldots,[h_{k+1}]$ be the prefix with $k+1$ states of a simple path in $\RL$ from $[\epsilon]$ to $[h]$.  For every $1\le i\le k+1$, we
	have $|h_{i}|=i-1$, and we define $H_{k}=\{h_{1},\ldots,h_{k}\}$.
	\hfill \qed
\end{proof}

Consider a language $L\subseteq \Sigma^\omega$ of infinite words. Here, the language complementing $L$ is ${\it comp}(L)=\Sigma^\omega \setminus L$. A finite word $x \in \Sigma^*$ is a {\em bad prefix\/} for $L$ if for every $y \in \Sigma^\omega$, we have that $x \cdot y \not \in L$. That is, $x$ is a bad prefix if all its extensions are words not in $L$. A language $L \subseteq \Sigma^\omega$ is a {\em safety\/} language if every word not in $L$ has a bad prefix. A language $L$ is a {\em co-safety\/} language if ${\it comp}(L)$ is safety. Equivalently, every word $w \in L$ has a {\em good prefix}, namely a prefix  $x \in \Sigma^*$ such that for every $y \in \Sigma^\omega$, we have that $x \cdot y \in L$.

\subsection{Transducers and Realizability}

Consider two finite alphabets $\Sigma_I$ and $\Sigma_O$.
For two words $x=x_1 \cdot x_2 \cdots \in \Sigma_I^\omega$ and $y = y_1 \cdot y_2 \cdots \in \Sigma_O^\omega$, we define $x \oplus y$ as the word in $(\Sigma_I \times \Sigma_O)^\omega$ obtained by merging $x$ and $y$. Thus, $x\oplus y = (x_1,y_1) \cdot (x_2,y_2) \cdots$.

A {\em $(\Sigma_I/\Sigma_O)$-transducer\/} models a finite-state system that generates letters in $\Sigma_O$ while interacting with an environment that generates letters in $\Sigma_I$. Formally, a $(\Sigma_I/\Sigma_O)$-transducer is $\T=\zug{\Sigma_I,\Sigma_O,\init,S,s_0,\rho,\tau}$, where $\init \in \{{\it sys,env}\}$ indicates who initiates the interaction -- the system or the environment, $S$ is a set of states, $s_0 \in S$ is an initial state, $\rho:S \times \Sigma_I \rightarrow S$ is a transition function, and $\tau:S \rightarrow \Sigma_O$ is a labeling function on the states.
Consider an input word $x=x_1 \cdot x_2 \cdots \in \Sigma_I^\omega$. The {\em run\/} of $\T$ on $x$ is the sequence $s_0,s_1,s_2 \ldots$ such that for all $j \geq 0$, we have that $s_{j+1} = \rho(s_j,x_{j+1})$. The {\em annotation of $x$ by $\T$}, denoted $\T(x)$, depends on $\iota$. If $\init={\it sys}$, then $\T(x) = \tau(s_0) \cdot \tau(s_1) \cdot \tau(s_2) \cdots \in \Sigma_O^\omega$. Note that the first letter in $\T(x)$ is the output of $\T$ in $s_0$. This reflects the fact that the system initiates the interaction. If $\init={\it env}$, then $\T(x) = \tau(s_1) \cdot \tau(s_2) \cdot \tau(s_3) \cdots \in \Sigma_O^\omega$. Note that now, the output in $s_0$ is ignored, reflecting the fact that the environment initiates the interaction. Then, the {\em computation\/} of $\T$ on $x$ is the word $x \oplus \T(x) \in (\Sigma_I \times \Sigma_O)^\omega$.

We say that a $(\Sigma_I/\Sigma_O)$-transducer is \emph{offline} if its behavior is independent of inputs from the environment. Formally, its transition function $\rho$ satisfies $\rho(s, x) = \rho(s, x')$ for all states $s \in S$ and input letters $x, x' \in \Sigma_I$. Note that an offline transducer has exactly one run, and it annotates all words by the same lasso-shaped word $u \cdot v^\omega$, with $u \in \Sigma_O^*$ and $v \in \Sigma_O^+$. We sometimes refer to general transducers as {\em online\/} transducers, to emphasize they are not offline.

Consider a $\omega$-regular language $L \subseteq (\Sigma_I \times \Sigma_O)^\omega$. We say that $L$ is {\em $(\Sigma_I/\Sigma_O)$-realizable by the system\/} if there exists a $(\Sigma_I/\Sigma_O)$-transducer $\T$ with $\init={\it sys}$  all whose computations are in $L$. Thus, for every $x\in \Sigma_I^\omega$, we have that $x \oplus \T(x)\in L$. We then say that $\T$ {\em $(\Sigma_I/\Sigma_O)$-realizes} $L$.
Then, $L$ is {\em $(\Sigma_O/\Sigma_I)$-realizable by the environment\/} if there exists a $(\Sigma_O/\Sigma_I)$-transducer $\T$ with $\init={\it env}$  all whose computations are in $L$.
When $\Sigma_I$ and $\Sigma_O$ are clear from the context, we omit them.
When the language $L$ is $\omega$-regular, realizability reduces to deciding a game with a regular winning condition. Then, by determinacy of games and due to the existence of finite-memory winning strategies \cite{BL69}, we have the following.

\begin{proposition}\label{determinacy}
For every $\omega$-regular language $L \subseteq (\Sigma_I \times \Sigma_O)^\omega$, exactly one of the following holds.
\begin{enumerate}
\item
$L$ is $(\Sigma_I/\Sigma_O)$-realizable by the system.
\item
$\comp(L)$ is $(\Sigma_O/\Sigma_I)$-realizable by the environment.
\end{enumerate}
\end{proposition}

\section{Proving and Refuting Bounds on DFAs}
\label{sec:recognizability}

Consider a regular language $L \subseteq \Sigma^*$ and a bound $k \geq 1$.
We view the problem of deciding whether $L$ can be recognized by a $k$-DFA as the problem of
deciding a turn-based two-player game between Refuter and Prover. In
each round in the game, Prover provides a letter from a set $[k]=\{1,2,\ldots,k\}$ that describes the state space of a DFA for $L$ that Prover claims to exist, and Refuter responds with a letter in $\Sigma \cup \{\#\}$, for a special reset letter $\# \not \in \Sigma$. Thus, during the interaction, Prover generates a word $y \in [k]^\omega$ and Refuter generates a word $x \in (\Sigma \cup \{\#\})^\omega$. The word $x$ describes an infinite sequence of words in $\Sigma^*$, separated by $\#$'s, and the word $y$ aims to describe runs of the claimed DFA on the words in the sequence.

Below we formalize this intuition. Let $\Sigma' = \Sigma \cup \{\#\}$, for a letter $\# \not \in \Sigma$. Consider a (finite or infinite) word $w = x \oplus y \in (\Sigma' \times [k])^* \cup (\Sigma' \times [k])^\omega$. Let $x=x_1 \cdot x_2 \cdots$ and $y = y_1 \cdot y_2 \cdots$. We say that $w$ is {\em legal\/} if the following two conditions hold:
\begin{enumerate}
\item For all $1 \leq j < |w|$ with $x_j = \#$, we have $y_{j+1} = y_1$.
\item There exists a function $\delta : [k] \times \Sigma \rightarrow [k]$ such that $y_{j+1}= \delta(y_j, x_j)$ for all $1 \leq j < |w|$ with $x_j \in \Sigma$.
\end{enumerate}

The first condition ensures that Prover starts all runs in the same state $y_1 \in [k]$, which serves as the initial state in her claimed DFA.
The second condition ensures that there exists a deterministic transition relation that Prover follows in all her transitions.

A word $w$ being legal guarantees that Prover follows some $k$-DFA. We now add conditions on $w$ in order to guarantee that this DFA recognizes $L$. Consider a position $1 \leq j < |w|$. Let $\#(j) = \max \{j' : (j' < j \mbox{ and } x_{j'}=\#) \mbox{ or } j' = 0\}$ be the last position before $j$ in which Refuter generates the reset letter $\#$ (or $0$, if no such position exists). When the interaction is in position $j$, we examine the word $w^j$ that starts at position $\#(j) + 1$ and ends at position $j-1$. Thus, $w^j = x_{\#(j)+1} \cdot x_{\#(j)+2} \cdots x_{j-1} \in \Sigma^*$. The run that Prover suggests to $w^j$ is then $y_{\#(j)+1}, y_{\#(j)+2},\ldots,y_j$, and we say that $y$ {\em maps\/} $w^j$ to $y_{j}$. When $y$ is clear from the context, we also say that Prover maps $w^j$ to $y_j$. Note that if $j_1$ and $j_2$ are such that $w^{j_1}=w^{j_2}$, then $w$ being legal ensures that  $w^{j_1}$ and $w^{j_2}$ are mapped to the same state. Now, we say that $w = x \oplus y \in (\Sigma' \times [k])^* \cup (\Sigma' \times [k])^\omega$ {\em agrees with $L$\/} if there exists a set $F \subseteq [k]$ such that for all $1 \leq j < |w|$, Prover maps $w^j$ to an element in $F$ iff $w^j \in L$.

\begin{remark}
Note that a word $w$ agrees with $L$ iff $w$ agrees with $\comp(L)$. Indeed, our definition of agreement with $L$ only guarantees we can define an acceptance condition on top of the claimed $k$-DFA for either $L$ and $\comp(L)$. Since these DFAs dualize each other, they have the same index, and so it makes sense not to distinguish between them in our study.
\hfill \qed \end{remark}

\begin{example}
Let $\Sigma=\{a,b\}$ and $k=2$.
An interaction between Prover and Refuter may generate the prefix of a computation in $(\{a,b,\#\} \times \{1,2\})^\omega$ described in Table \ref{xoplusy example}. Note that while $w$ fixes $\delta(1,a)$, $\delta(2,a)$, and $\delta(2,b)$, it does not fix $\delta(1,b)$.
\
\begin{table}[htb]
	\begin{center}
		\begin{tabular}{rccccccccccc}
		$w=x \oplus y=$ & $(a,1)$ & $(b,2)$ & $(\#,2)$ & $(a,1)$ & $(a,2)$ & $(a,1)$ & $(b,2)$ & $(\#,2)$ & $(\#,1)$ & $(a,1)$ & $(a,2)$ \\
		$j=$ & $1$ & $2$ & $3$ & $4$ & $5$ & $6$ & $7$ & $8$ & $9$ & $10$ & $11$  \\
		$\#(j)=$ & $0$ & $0$ & $0$ & $3$ & $3$ & $3$ & $3$ & $3$ & $8$ & $9$ & $9$  \\
$w^j=$ & $\epsilon$ & $a$ & $ab$ & $\epsilon$ & $a$ & $aa$ & $aaa$ & $aaab$ & $\epsilon$ & $\epsilon$ & $a$ \\
		\end{tabular}
	\end{center}
	\caption{$x \oplus y$ and its analysis.}
	\label{xoplusy example}
\end{table}

 In Figure~\ref{aw fig} we describe two possible DFAs induced by $w$ and the two possible choices for $\delta(1,b)$.

\begin{figure}[htb]
	\begin{center}
	\includegraphics[width=7cm]{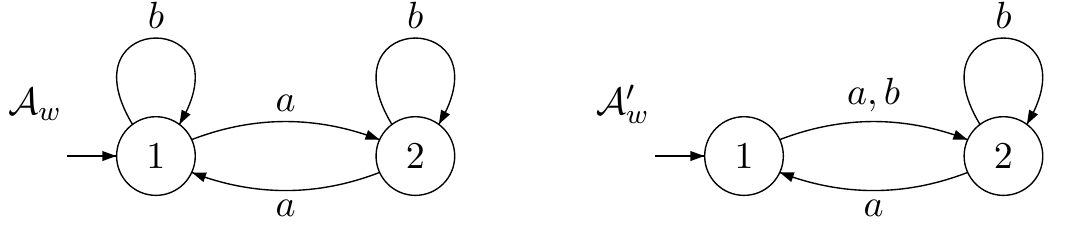}
	\end{center}
	\caption{The DFAs $\A_w$ and $\A'_w$ induced by $w$.}
	\label{aw fig}
\end{figure}

Consider the language $L_1 \subseteq \{a,b\}^*$ of all words with an even number of $a$'s. Then, $w$ agrees with $L_1$, since there is $F = \{1\}$ and all $w^j$ with an even number of $a$'s are mapped to $F$. However, if we consider the language $L_2 \subseteq \{a,b\}^*$ of all words with an even number of $b$'s, there is no $F$ witnessing that $w$ agrees with $L_2$. Clearly, any $F \subseteq \{2\}$ is not a witness, since $\epsilon \in L_2$, but $1 \notin F$. Moreover, $F = \{1,2\}$ cannot be a witness, since $ab \notin L_2$, and $F = \{1\}$ is also ruled out, since $a \in L_2$. Thus $w$ does not agree with $L_2$.
\hfill \qed \end{example}

The language $\DFAreal(L,k) \subseteq (\Sigma' \times [k])^\omega$ of words with correct annotations is then
$\DFAreal(L,k) = \{ x \oplus y \in (\Sigma' \times [k])^\omega : x \oplus y \text{ is legal and agrees with } L  \}.$
Then, $\NoDFAreal(L,k)$ is the language of words with incorrect annotations, thus $\NoDFAreal(L,k) = \comp(\DFAreal(L,k))$.

By \Cref{determinacy}, we have the following:

\begin{proposition}\label{dfa:recognisability}
Consider a language $L \subseteq \Sigma^*$. Exactly one of the following holds:
\begin{itemize}
\item
$L$ can be recognized by a $k$-DFA, in which case $\DFAreal(L,k)$ is $(\Sigma'/[k])$-realizable by the system.
\item
$L$ cannot be recognized by a $k$-DFA, in which case $\NoDFAreal(L,k)$ is $([k]/\Sigma')$-realizable by the environment.
\end{itemize}
\end{proposition}

By Proposition \ref{dfa:recognisability}, the language $\DFAreal(L,k)$ is $(\Sigma'/[k])$-realizable by the system whenever $k \geq \indx(L)$. Moreover, as we argue below, a $(\Sigma'/[k])$-transducer $\T$ that realizes $\DFAreal(L,k)$ induces a $k$-DFA for $L$. To see this, consider the word $x \in (\Sigma')^*=
w_1 \cdot \# \cdot w_2 \cdots \# \cdot w_{|\Sigma|^k} \cdot \#$ obtained by concatenating all words $w_i \cdot \# \in \Sigma^k \cdot \#$ in some order. Since every transition in a $k$-DFA is reachable by traversing a word of length at most $k-1$, the computation of $\T$ on $x$ must commit on all the transitions in a transition function $\delta:[k] \times \Sigma \rightarrow [k]$, and must also induce a single classification of the states in $[k]$ to accepting and rejecting.
Note also that if $k > \indx(L)$, the transducer may induce several different DFAs for $L$.

By Proposition \ref{dfa:recognisability}, we also have that the language $\NoDFAreal(L,k)$ is $([k]/\Sigma')$-realizable by the environment whenever $k < \indx(L)$. A $([k]/\Sigma')$-transducer that realizes $\NoDFAreal(L,k)$ is termed an {\em $(L,k)$-refuter}.

\section{Certifying Bounds on Recognizability}
\label{sec cert}

Recall that $\DFAreal(L,k)$ contains exactly all words that are legal and agree with $L$. Accordingly, if a word
$x \oplus y \in (\Sigma' \times [k])^\omega$ is not in $\DFAreal(L,k)$, it contains a violation of legality or agreement with $L$, and thus has a bad prefix for $\DFAreal(L,k)$. Formally, we define the language $\violate(L,k) \subseteq (\Sigma' \times [k])^*$ of words that include a violation of legality or agreement with $L$ as follows.
\[\begin{array}{ll}
\violate(L,k) =  \{ x \oplus y : & \text{there is } j \geq 1 \text{ such that } x_j=\# \text{ and } y_{j+1} \neq y_1, \mbox{ or} \\
                & \text{there are } j_1, j_2 \geq 1 \text{ such that }\\
                 & y_{j_1}=y_{j_2}, x_{j_1}=x_{j_2}, \text{ and } y_{{j_1}+1} \neq y_{{j_2}+1},  \\
                & \mbox{or } w^{j_1} \in L, w^{j_2} \notin L \text{ and } y_{{j_1}} = y_{{j_2}}\}.
\end{array}\]
Note that while all the words in $\violate(L,k)$ are bad prefixes for $\DFAreal(L,k)$, there are bad prefixes for $\DFAreal(L,k)$ that are not in $\violate(L,k)$. For example, if $L=\{a^{2n} : n \geq 0\}$, then the word $(a,1)$ is a bad prefix for $\DFAreal(L,1)$, as both $(a,1)(a,1)$ and $(a,1)(\#,1)$, which are the only possible extensions of $(a,1)$ by a single letter, are in $\violate(L,1)$, yet $(a,1)$ itself is not in $\violate(L,1)$. Formally, using the terminology of \cite{KV01d}, the language $\violate(L,k)$ contains all the {\em informative bad prefixes\/} of $\DFAreal(L,k)$, namely these that contain an explanation to the prefix being bad. Since every infinite word not in $\DFAreal(L,k)$ has a bad prefix in $\violate(L,k)$, then restricting attention to bad prefixes in $\violate(L,k)$ is appropriate in the context of certificates. Also, as we discuss in Remark~\ref{violate is close}, a bad prefix of $\DFAreal(L,k)$ that is not informative can be made informative by concatenating to it any letter in $\Sigma' \times [k]$.

\begin{remark}\label{violate is close}
Surprisingly, extending a bad prefix of $\DFAreal(L,k)$ by any letter of $\Sigma' \times [k]$ transforms it to an informative bad prefixes, i.e., makes it an element of $\violate(L,k)$: Let $w = (x_1, y_1) \cdots (x_n, y_n)$ be a bad prefix. In particular, we have $(x_1, y_1) \cdots (x_n, y_n) \cdot (\#, y_{n+1}) \cdot (\#, y_1)^\omega \notin \DFAreal(L,k)$ for all $y_{n+1} \in [k]$. Since continuing a word with $(\#, y_1)$ after a preceding $\#$ does not impact legality or agreement with $L$, the word $w' = (x_1, y_1) \cdots (x_n, y_n) \cdot (\#, y_{n+1})$ must include a violation of legality or agreement with $L$ and thus $w' \in \violate(L,k)$. Lastly, since the definition of $\violate(L,k)$ does not refer to the $\Sigma'$ component of the last letter read, we can replace $\#$ by any letter of $\Sigma'$ and thus have shown that any letter of $\Sigma' \times [k]$ transforms a bad prefix to an informative bad prefix.
\hfill \qed \end{remark}

Refuting recognizability of $L$ by a $k$-DFA, we consider two approaches. In the first, we consider the interaction of Prover with an offline $(L,k)$-refuter. Such a refuter has to  generate a word $x \in (\Sigma')^*$ such that for all $y \in [k]^{|x|}$, we have that $x \oplus y \in \violate(L,k)$. We call $x$ a {\em universal informative bad prefix} (see  \cite{KW12} for a study of bad prefixes for safety languages in an interactive setting). In the second approach, we consider the interaction of Prover with an online $(L,k)$-refuter. There, the goal is to associate every sequence $y \in [k]^\omega$ that is generated by Prover with a sequence $x \in (\Sigma')^\omega$ such that $x \oplus y$ has a prefix in $\violate(L,k)$. In Sections~\ref{na ref} and~\ref{a ref} we compare the two approaches in terms of the length of the certificate (namely the word in $\violate(L,k)$) that they generate.

\subsection{Certification with Offline Refuters}
\label{na ref}

Recall that a word $x \in (\Sigma')^*$ is a {\em universal informative bad prefix\/} for $\DFAreal(L,k)$ if for all $y \in [k]^{|x|}$, we have that $x \oplus y\in \violate(L,k)$.

\begin{theorem}\label{thm:non-adapative}
Consider a regular language $L \subseteq \Sigma^*$ and let $N = \indx(L)$. For every $k < N$, the length of a shortest universal informative bad prefix for $\DFAreal(L,k)$ is at most $O(k^2 \cdot N)$. This bound is tight: There is a family of regular languages $L_1,L_2, \ldots$ such that for every $n \geq 1$, the length of a shortest universal informative bad prefix for $\DFAreal(L_n, N_n - 1)$ is $\Omega(N_n^3)$, where $N_n=\indx(L_n)$.
\end{theorem}

\begin{proof}
We start with the upper bound and construct, for every $k < N$, a universal informative bad prefix for $\DFAreal(L,k)$ of length $O(k^2 \cdot N)$.

Let $H = \{h_1,\ldots, h_{k + 1}\}$ be representatives of $k + 1$ distinct Myhill-Nerode classes. Since $k < N$, such a set $H$ exists. Moreover, by Lemma~\ref{head bound}, we can assume that $|h_i| \leq k$, for all $1 \leq i \leq k+1$. For each pair $\zug{h_i, h_j}$, there is a distinguishing tail $t_{i,j}$ of length at most $N$.
Let  $x$ be the concatenation of all words of the form $h_i \cdot t_{i,j} \cdot \#$ and $h_j \cdot t_{i,j} \cdot \#$, for all pairs. There are $k \cdot (k+1)$ such words, each of length at most $k+N+1$, so $|x| \leq (k+N+1)\cdot k \cdot (k+1)$, which is $O(k^2\cdot N)$.
Below we prove that $x$ is a universal informative bad prefix.

	Let $y \in [k]^{|x|}$. For every $h_i\in H$, the subword $\# \cdot h_i$ appears in $x$, so $y$ maps $h_i$ to some element in $[k]$. By the pigeonhole principle, there are two distinct words $h_i$ and $h_j$ such that $y$ maps both words to the same element. If $x \oplus y$ is legal, the transitions are consistent, so $y$ maps both $h_i \cdot t_{i,j}$ and $h_j \cdot t_{i,j}$ to the same state. Then, however, as exactly one of $h_i \cdot t_{i,j}$ and $h_j \cdot t_{i,j}$ is in $L$, there is no $F\subseteq [k]$ that satisfies the condition of agreement with $L$, and so $x \oplus y \in \violate(L,k)$. Hence, $x$ is a universal informative bad prefix for $\DFAreal(L,k)$.

For a matching lower bound, we describe a family of regular languages $L_1,L_2, \ldots$ such that for every $n \geq 1$, the length of a shortest universal informative bad prefix for $\DFAreal(L_n, N_n - 1)$ is $\Omega(N_n^3)$, where $N_n=\indx(L_n)$. For $n \geq 1$, let $\Sigma_n = \{a, b_1, \ldots, b_n\}$ and consider the language $L_n = \{a^n b_i^2 : 1\leq i \leq n\}$. Let $\A_n$ be a minimal DFA for $L_n$. For example, $L_3 = \{aaab_1b_1, aaab_2b_2, aaab_3b_3\}$, and the DFA $\A_3$ for $L_3$ appears in Figure \ref{fig:dfa_L3}.

\begin{figure}
	\begin{center}
	\vspace{-2mm}
	\includegraphics[width=9cm]{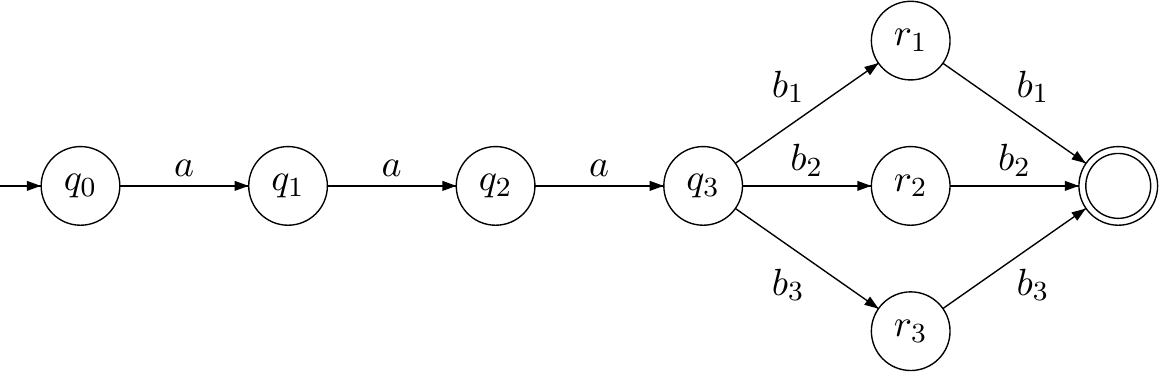}
	\end{center}
	\caption{A DFA for $L_3$.}
	\label{fig:dfa_L3}
\end{figure}

	It is easy to see that $\indx(L_n)=N_n = 2n + 3$, corresponding to (see Figure~\ref{fig:dfa_L3}) $n + 1$ states $q_0, \ldots, q_n$, $n$ states $r_1, \ldots, r_n$, an accepting state, and a rejecting sink, which we omit from the figure.

	Let $k = N_n - 1$, and consider some prefix $x \in (\Sigma')^*$. For $1 \leq i \neq j \leq n$, the words $a^n b_i$ and $a^n b_j$ belong to different Myhill-Nerode classes, corresponding to the states $r_i$ and $r_j$, respectively. The distinguishing tails are $b_i$ and $b_j$. We claim that if $x$ does not contain the subword $a^n b_i b_j$ or $a^n b_j b_i$, then there is $y \in [k]^{|x|}$ such that $x \oplus y \not \in \violate(L,k)$, and so $x$ is not a universal informative bad prefix for $\DFAreal(L,k)$. To see this, consider the word $y\in [k]^{|x|}$ constructed by following the DFA obtained from $\A_n$ by merging the states $r_i$ and $r_j$. We can choose $x' = \#^\omega$ and $y'\in [k]^\omega$ such that $y'_1$ is consistent with the transitions in $y$, and $y'_j = y_1$ for all $j\geq 2$. If $x=y=\epsilon$, we can choose $x' \oplus y' = (\#,1)^\omega$. Then, $(x\cdot x')\oplus (y\cdot y') \in \DFAreal(L,k)$, and so $x\oplus y$ is not an informative bad prefix for $\DFAreal(L,k)$.

	Hence, if $x$ is a universal informative bad prefix for $\DFAreal(L,k)$, then for every $1 \leq i \neq j \leq n$, it contains the subwords $a^n b_i b_j$ or $a^n b_j b_i$, which are of length $n + 2$. There are $n \cdot (n-1)/2$ such subwords and they are disjoint. Therefore, $|x| \geq (n+2)\cdot n \cdot (n-1)/2$, which is $\Omega(N_n^3)$.
\hfill \qed \end{proof}

\subsection{Certification with Online Refuters}
\label{a ref}

We now consider refuters that take Prover's choices into account when outputting letters. We show that this capability allows an interactive refuter to win in fewer rounds than an offline refuter.

\begin{theorem}\label{interactive refuters bad prefix}
Consider a regular language $L \subseteq \Sigma^*$ and let $N = \indx(L)$. For every $k < N$, there exists an $(L,k)$-refuter that generates a word in $\violate(L,k)$ within $O(k^2 + N)$ rounds. This bound is tight: There is a family of regular languages $L_1,L_2,\ldots$ such that for every $n\geq 1$, every $(L,k)$-refuter needs at least $\Omega(N_n^2)$ rounds to construct a word in $\violate(L_n,N_n-1)$, where $N_n=\indx(L_n)$.
\end{theorem}

\begin{proof}
We start with the upper bound, by describing a winning strategy. As in the offline case, let $H=\{h_1, \ldots, h_{k + 1}\}$ be representatives of distinct Myhill-Nerode classes, each of length at most $k$. Unlike the offline case, where Refuter outputs all pairs of heads and distinguishing tails, here a single pair suffices to achieve the same effect. Refuter starts the interaction by outputting $h_1 \cdot \# \cdots \# \cdot h_{k + 1} \cdot \#$. By the pigeonhole principle, there are distinct words $h_i$ and $h_j$ that are mapped by Prover to the same state. Refuter then outputs $ h_i \cdot t_{i,j}\cdot \# \cdot h_j \cdot t_{i,j} \cdot \#$. If Prover does not violate the conditions of legality, it maps $h_i \cdot t_{i,j}$ and $h_j \cdot t_{i,j}$ to the same state. Exactly one of them is in $L$, so there is no $F\subseteq [k]$ that can satisfy agreement with $L$, and so the generated word is in $\violate(L,k)$. We now analyze its length. Recall that Refuter first outputs $k + 1$ words of length at most $k$ each, separated by $\#$'s, and then two words of length at most $k+N$ each, again separated by $\#$. Thus, the length of the prefix is $k(k + 1) + 2(k+N) + k + 3$, which is $O(k^2 + N)$.

For a matching lower bound, we describe a family of regular languages $L_1,L_2\ldots$ such that for every $n\geq 1$, every refuter needs at least $\Omega(N_n^2)$ rounds to generate a word in $\violate(L_n,N_n-1)$, where $N_n=\indx(L_n)$. Consider the DFA $\A_n$ from the offline lower bound, again with $k = N_n - 1$. We claim that $\Omega(N_n^2)$ rounds are required to generate a word in $\violate(L_n,N_n-1)$.

Let $x\in (\Sigma')^*$ be the word generated by Refuter. Assume there exists $1\leq i\leq n$ such that the subword $a^n b_i$ does not appear in $x$. The state corresponding to $a^n b_i$ is $r_i$. Hence, Prover can follow the DFA obtained by removing the state $r_i$ from $\A_n$ without violating legality or agreement with $L$. Therefore, in order to guarantee a generation of a word in $\violate(L_n,N_n-1)$, Refuter  must output all the words $a^n b_1,\ldots , a^n b_n$ in some order. Each of these $n$ words has length $n + 1$, and they are disjoint. Their total length is therefore at least $n\cdot (n + 1)$, which is $\Omega(N_n^2)$.
\hfill \qed \end{proof}

\begin{remark}
{\bf Fixed alphabet}
	In the proofs of Theorems \ref{thm:non-adapative} and~\ref{interactive refuters bad prefix}, we use languages $L_n$ over an alphabet $\Sigma_n$ that depends on $n$. By replacing the letters $b_1,\ldots,b_n$ by words in $\{a,b\}^{\lfloor \log n \rfloor}$, one gets languages over the fixed alphabet $\Sigma=\{a,b\}$ that exhibit the claimed lower bounds for both online and offline refuters.\hfill \qed
\end{remark}

\subsection{Optimal Survival Strategies for Provers}

Assume that $L$ is not $k$-DFA recognizable. Then, there is an $(L,k)$-refuter, and Refuter is going to win a game against Prover and generate a word in $\violate(L,k)$. Suppose that Prover aims at prolonging the interaction. It is tempting to think that the following greedy strategy is optimal for such an objective: Prover follows the transitions of $\RL$. If $k < \indx(L)$, then Prover may be forced to deviate from $\RL$ and make a ``mistake'', namely choose to output one of the $k$ states that have already been exposed. Using this strategy, Prover can prolong the game at least until $k + 1$ different states are exposed. The following example shows that this strategy is not necessarily best at prolonging the game as long as possible (no matter how clever the choice when a ``mistake" is forced is).

\begin{example}
	For $n \geq 1$, consider the language $L_n = \{w : w_1 = b\text{ or }|w|=n\}$ over $\Sigma = \{a,b\}$	.
	\begin{figure}
		\begin{center}
			\includegraphics[width=10cm]{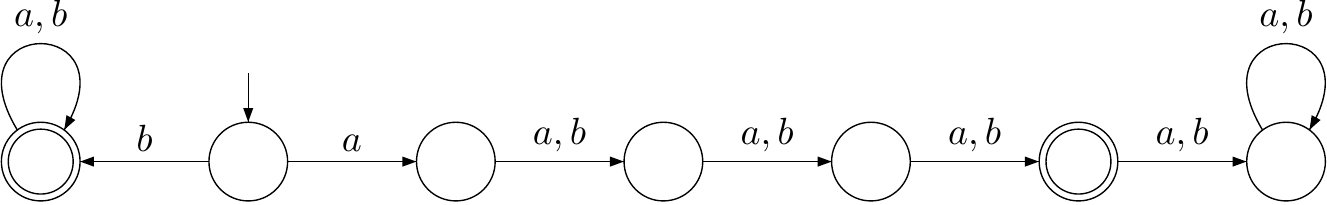}
		\end{center}
		\caption{A DFA for $L_4$.}
		\label{fig:dfa_L4}
	\end{figure}

	Denote the canonical DFA for $L_n$ by $\A_n$. For example, $\A_4$ appears in Figure \ref{fig:dfa_L4}.
	The number of states is $n + 3$. Let $k = n + 2$. We claim that if Prover follows $\A_n$, then $x = a^{n + 1} \cdot \# \cdot  b \cdot a$ induces a bad prefix $x \oplus y$ for $L_n$. The first word, $a^{n + 1}$, forces Prover to expose all $n + 2$ states, after which it cannot have an accepting sink. Then, $\# \cdot b$ forces Prover to choose an existing state instead of an accepting sink. To prolong the game, it chooses the only accepting state, and then the last $a$ ends the game. The number of rounds needed to win against this prover is at most $n + 4$.

	Prover can do better than $n + 4$ rounds. Let $L'_n = \{w : w_1 = b\text{ or } (w_1 = a\text{ and } |w| = 0 \bmod n)\}$ and let $\B_n = \mathcal{R}_{L'_n}$ be the minimal DFA for $L'_n$. For example, $\B_4$ appears in Figure \ref{fig:dfa_B4}.
	The shortest possible word length on which $\A_n$ and $\B_n$ disagree is $2n$ (for example, $a^{2n} \in L_n' \setminus L_n$), which is better than $n + 4$. It can be shown that an $(L_n,k)$-refuter can generate a bad prefix for $L_n$ in at most $2n$ rounds against all provers, so $\B_n$ is optimal in that sense.
	\begin{figure}
		\begin{center}
			\includegraphics[width=8cm]{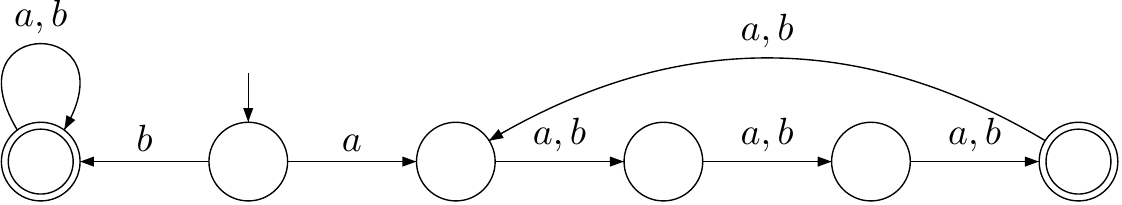}
		\end{center}
		\caption{The DFA $\B_4$.}
		\label{fig:dfa_B4}
	\end{figure}
\hfill \qed \end{example}

\section{Bounds on DFA Separation}

Consider three languages $L_1, L_2,L \subseteq \Sigma^*$. We say that $L$ is a {\em separator} for $\zug{L_1,L_2}$ if $L_1 \subseteq L$ and $L \cap L_2 = \emptyset$. Equivalently, $L_1 \subseteq L \subseteq \comp(L_2)$.
For $k \geq 1$, we say that a pair of languages $\zug{L_1,L_2}$ is {\em $k$-DFA-separable} iff there is a $k$-DFA $\A$ such that $L(\A)$ separates $\zug{L_1,L_2}$. We extend the definition to DFAs and say that two DFAs $\A_1$ and $\A_2$ are separated by a DFA $\A$, if their languages are separated by $L(\A)$.

In this section we study refuting and certifying bounds on DFA separation. We first give proofs that deciding (strict and non-strict) $k$-DFA-separability, is NP-complete.
The problem being NP-hard suggests that there is no clean theory of equivalence classes that is the base for offline certification. We continue and describe interactive certification protocol for $k$-DFA-separability.

\subsection{Hardness of Separation}

The following Theorem \ref{sep np} is considered by the literature (e.g., \cite{Nei12}) to be a consequence of \cite{Pfl73}. Since we also investigate the strict-separation case and there is a progression of techniques, we describe below an alternative and explicit proof.

\begin{theorem}
\label{sep np}
	Given DFAs $\A_1$ and $\A_2$, and a bound $k \geq 1$, deciding whether $\zug{\A_1,\A_2}$ is $k$-DFA-separable is NP-complete.
\end{theorem}

\begin{proof}
Membership in NP is easy, as given a candidate separator $\A$ of size $k$, we can verify that $L(\A_1) \subseteq L(\A)$ and $L(\A) \cap L(\A_2) = \emptyset$ in polynomial time. Note that if $k \geq \indx(L(\A_1))$, then $\zug{\A_1,\A_2}$ is $k$-DFA-separable by $\A_1$. Thus, we can assume that $k <  \indx(L(\A_1))$, and so membership in NP applies also for the case $k$ is given in binary.

	For NP-hardness, we reduce from  the {\em DFA identification\/} problem. Recall that there, given sets $S_1,S_2 \subseteq \Sigma^*$ of positive and negative words, and a bound $k \geq 1$, we seek a $k$-DFA that accepts all words in $S_1$ and no word in $S_2$. By \cite{Gol78}, DFA identification is NP-complete. Given $S_1$, $S_2$, and $k$, our reduction constructs DFAs $\A_1$ and $\A_2$ such that $L(\A_1) = S_1$ and $L(\A_2) = S_2$. Clearly, a $k$-DFA solves the DFA identification problem for $S_1$, $S_2$, and $k$, iff it solves the $k$-DFA-separation of $\A_1$ and $\A_2$.

	Constructing a DFA $\A_S$ such that $L(\A_S)=S$, for some finite set $S \subseteq \Sigma^*$ can be done in polynomial time, by traversing prefixes of words in $S$. Formally, we define $\A_S=\zug{\Sigma,Q,q_0,\delta,F}$, where $Q=\{w : w \mbox{ is a prefix of a word in }S\}$, $q_0=\epsilon$, and for all $w \in Q$ and $\sigma \in \Sigma$, we have that $\delta(w,\sigma)=w \cdot \sigma$ if $w \cdot \sigma \in S$, and $\delta(w,\sigma)$ is undefined otherwise. Finally, $F=S$. It is easy to see that $L(\A_S)=S$ and that $|\A_S| \leq \sum_{w\in S} |w|$.
\hfill \qed \end{proof}

Consider three languages $L_1, L_2,L \subseteq \Sigma^*$. We say that $L$ is a {\em strict separator} for $\zug{L_1,L_2}$ if $L_1 \subset L$, $L \cap L_2 = \emptyset$, and $L \cup L_2 \subset \Sigma^*$. Equivalently, $L_1 \subset L \subset \comp(L_2)$. For $k \geq 1$, we say that a pair of languages $\zug{L_1,L_2}$ is {\em $k$-DFA-strictly-separable} iff there is a $k$-DFA $\A$ such that $L(\A)$ strictly separates $\zug{L_1,L_2}$. Again, we extend the definition to DFAs.

\begin{theorem}\label{ssep np}
Given DFAs $\A_1$ and $\A_2$, and a bound $k \geq 1$, deciding whether $\zug{\A_1,\A_2}$ is $k$-DFA-strictly-separable is NP-complete.
\end{theorem}

\begin{proof}
We start with membership in NP. As in the proof of Theorem \ref{sep np}, a witness $k$-DFA $\A$ can be checked in polynomial time. However, if $k$ is given in binary and greater than $\indx(L(\A_1))$ and $\indx(L(\A_2))$, we cannot base a separator on $\A_1$ or $\A_2$. We fill this gap by showing that if a DFA strictly separates $\zug{\A_1,\A_2}$, then there also exists one that is polynomial in $|\A_1|$ and $|\A_2|$.

Assume that $\zug{\A_1,\A_2}$ are strictly separable, and let $T = \comp (L(\A_1) \cup L(\A_2))$. Note that  $\zug{\A_1,\A_2}$ being strictly separable implies that $|T| >1$. Let $\A_T$ be a minimal DFA for $T$. Note that $|\A_T| \leq |\A_1| \cdot |\A_2|$. Consider a word $w\in T$ that is accepted along a simple path in $\A_T$. Thus, $|w|$ is polynomial in $|\A_T|$. Consider a DFA $\A^w_1$ with $L(\A_1^w) = L(\A_1) \cup \{w\}$. Note that $|\A_1^w|$ is polynomial in $|\A_1|$ and $|w|$. It is not hard to see that $\A_1^w$ is a strict separator for $\zug{\A_1,\A_2}$. Indeed, $L(\A_1^w)$ strictly contains $L(\A_1)$, it is contained in $\comp(L(\A_2))$, and as $|T| >1$, the latter containment is strict. Hence,  $\zug{\A_1,\A_2}$ are strictly separable by a DFA that is polynomial in $|\A_1|$ and $|\A_2|$.

For NP-hardness, we describe a reduction from $k$-DFA-separability, proved to be NP-hard in Theorem~\ref{sep np}. Consider two DFAs $\A_1$ and $\A_2$ over $\Sigma$, and assume that $0 \not \in \Sigma$. Assume also that $L(\A_1),L(\A_2) \neq \emptyset$, and that $L(\A_1),L(\A_2)$ are finite, and thus have rejecting sinks. Clearly, $k$-DFA-separability is NP-hard also in this case. Let $\A'_1$ and $\A'_2$ be DFAs obtained from $\A_1$ and $\A_2$ by extending the alphabet to $\Sigma \cup \{0\}$ and adding a transition labeled $0$ from every state to the rejecting sink. Note that $L(\A'_1)=L(\A_1)$ and $\comp(L(\A'_2))=(\Sigma \cup \{0\})^* \setminus L(\A_2)$. We prove that for every $k \geq 1$, we have that $\zug{\A_1,\A_2}$ is $k$-DFA-separable iff $\zug{\A'_1,\A'_2}$ is $k$-DFA-strictly-separable.

Assume that there is a $k$-DFA $\A=\zug{\Sigma,Q,\delta,q_0,F}$ that separates $\zug{\A_1, \A_2}$. Let $\A'$ be the $k$-DFA obtained from $\A$ by extending the alphabet to $\Sigma \cup \{0\}$, and adding a transition labeled $0$ from every state to $q_0$. It is easy to see that $L(\A') = (\Sigma^* \cdot 0)^*\cdot L(\A)$, and so
$L(\A) \subseteq L(\A')$. Also, whenever $L(\A)$ is not empty, this containment is strict. Indeed, each word $w \in L(\A)$ induces the word $0 \cdot w \in L(\A') \setminus L(\A)$. Hence, as $\emptyset \neq L(\A_1) \subseteq L(\A)$, we have that $L(\A'_1) \subset L(\A')$. In addition, as $L(\A) \cap L(\A_2) = \emptyset$, then clearly $L(\A') \cap L(\A'_2) =\emptyset$. Moreover, as $L(\A_2)\neq \emptyset$, there is a word $w\in L(\A_2)$. Then, $w \not \in L(\A)$ and so $w\cdot 0 \cdot w\notin L(\A')$. In addition, $w\cdot 0 \cdot w\notin L(\A'_2)$. Thus, $L(\A')\cup L(\A'_2)\subset (\Sigma\cup\{0\})^*$, and we are done.

For the other direction, assume there is a $k$-DFA $\A'$ that strictly separates $\zug{\A'_1,\A'_2}$. Consider the $k$-DFA $\A$ obtained from $\A'$ by removing all transitions labeled $0$ and changing the alphabet to $\Sigma$. Every word in $L(\A_1)$ is also in $L(\A')$, and it does not contain $0$. So, $L(\A_1)\subseteq L(\A)$. Similarly, every word in $L(\A')$ that does not contain $0$ and is not in $L(\A'_2)$, is also not in $L(\A_2)$. Therefore, $L(\A) \cap L(\A_2)=\emptyset$.  Hence, $\A$ separates $\zug{\A_1,\A_2}$, and we are done.
\hfill \qed \end{proof}

The reduction described in the proof of Theorem~\ref{ssep np} can be used to prove NP-completeness also for {\em one-sided strict separation} problems. Formally, we have the following, which generalizes Conjecture~1 from \cite{GGS17}.

\begin{theorem}\label{hssep np}
Given DFAs $\A_1$ and $\A_2$, and a bound $k \geq 1$, the problems of deciding whether there exists a $k$-DFA $\A$ such that $L(\A_1) \subset L(\A) \subseteq \comp(L(\A_2))$ and whether there exists a $k$-DFA $\A'$ such that $L(\A_1) \subseteq L(\A') \subset \comp(L(\A_2))$ are NP-complete.
\end{theorem}

\begin{proof}
We start with the problem of deciding whether there exists a $k$-DFA $\A$ such that $L(\A_1) \subseteq L(\A) \subset \comp(L(\A_2))$.

Membership in NP is easy, as we can verify each containment in polynomial time. Note that, as in the proof of Theorem~\ref{sep np}, if $k \geq \indx(L(\A_1))$, then $\A_1$ is a witness. Thus, we can assume that $k < \indx(L(\A_1))$, and so membership in NP applies also for the case $k$ is given in binary.

For NP-hardness, we follow the same reduction from $k$-DFA-separability described in the proof of Theorem~\ref{ssep np}, and argue it is valid also for our problem.
Assume that there is a $k$-DFA $\A=\zug{\Sigma,Q,\delta,q_0,F}$ that separates $\zug{\A_1, \A_2}$. Let $\A'$ be the $k$-DFA obtained from $\A$ by extending the alphabet to $\Sigma \cup \{0\}$, and adding a transition labeled $0$ from every state to $q_0$. It is easy to see that $L(\A') = (\Sigma^* \cdot 0)^*\cdot L(\A)$, and so $L(\A) \subseteq L(\A')$. Therefore, we have $L(\A'_1)\subseteq L(\A')$. In addition, as $L(\A) \cap L(\A_2) = \emptyset$, then clearly $L(\A') \cap L(\A'_2) =\emptyset$. Moreover, as $L(\A_2)\neq \emptyset$, there is a word $w\in L(\A_2)$. Then, $w \not \in L(\A)$ and so $w\cdot 0 \cdot w\notin L(\A')$. In addition, $w\cdot 0 \cdot w\notin L(\A'_2)$. Thus, $L(\A')\cup L(\A'_2)\subset (\Sigma\cup\{0\})^*$, and we are done.

For the other direction, assume there is a $k$-DFA $\A'$ such that  $L(\A_1) \subseteq L(\A) \subset \comp(L(\A_2))$. Consider the $k$-DFA $\A$ obtained from $\A'$ by removing all transitions labeled $0$ and changing the alphabet to $\Sigma$. Every word in $L(\A_1)$ is also in $L(\A')$, and it does not contain $0$. So, $L(\A_1)\subseteq L(\A)$. Similarly, every word in $L(\A')$ that does not contain $0$ and is not in $L(\A'_2)$, is also not in $L(\A_2)$. Therefore, $L(\A) \cap L(\A_2)=\emptyset$.  Hence, $\A$ separates $\zug{\A_1,\A_2}$, and we are done.

Now, for the problem of deciding whether there exists a $k$-DFA $\A$ such that $L(\A_1) \subset L(\A) \subseteq \comp(L(\A_2))$, note that the latter condition is equivalent to $L(\A_2) \subseteq \comp(L(\A)) \subset \comp(L(\A_1))$, which is NP-complete by the above.
\end{proof}

\subsection{Certifying Bounds on Separation}

Consider two regular languages $L_1, L_2 \subseteq \Sigma^*$ and a bound $k \geq 1$. Certifying bounds on separation, we again consider a turn-based two-player game between Prover and Refuter. This time we are interested in whether $L_1$ and $L_2$ can be separated by a $k$-DFA. Consider a word $x \oplus y \in (\Sigma' \times [k])^\omega$. We say that $x\oplus y$ {\em agrees with\/} $\zug{L_1,L_2}$ if there exists $F \subseteq [k]$ such that for every $j \geq 1$, if $w^j \in L_1$, then Prover maps $w^j$ to $F$ and if $w^j \in L_2$, then Proven does not map $w^j$ to $F$.

Accordingly, we define the language $\DFAsep(L_1,L_2,k) \subseteq (\Sigma' \times [k])^\omega$ of words with correct annotations as follows:
\[\DFAsep(L_1,L_2,k)=\{x \oplus y : x \oplus y \mbox{ is legal and agrees with }\zug{L_1,L_2} \}.\]
Then, $\NoDFAsep(L_1,L_2,k)=\comp(\DFAsep(L_1,L_2,k))$ is the language of all words with incorrect annotations.
\begin{proposition}\label{dfa:separability}
Consider two regular languages $L_1, L_2 \subseteq \Sigma^*$ and $k \geq 1$. Exactly one of the following holds:
\begin{itemize}
\item
$\zug{L_1,L_2}$ is  $k$-DFA-separable, in which case $\DFAsep(L_1,L_2,k)$ is $(\Sigma'/[k])$-realizable by the system.
\item
$\zug{L_1,L_2}$ is not $k$-DFA-separable, in which case $\NoDFAsep(L_1,L_2,k)$ is $([k]/\Sigma')$-realizable by the environment.
\end{itemize}
\end{proposition}

A transducer that $([k]/\Sigma')$-realizes
$\NoDFAsep(L,k)$ is termed an $(L_1,L_2,k)$-refuter, and we seek refuters that generate short certificates.
As has been the case in Section~\ref{sec cert}, such a certificate is an informative bad prefix for $\DFAsep(L_1,L_2,k)$. Formally, we define the language $\violate(L_1,L_2,k) \subseteq (\Sigma' \times [k])^*$ of words that include a violation of legality or agreement with $L_1$ and $L_2$ as follows.
\[\begin{array}{ll}
\violatesep(L_1,L_2,k) =  \{ x \oplus y : & \text{there is } j \geq 1 \text{ such that } x_j=\# \text{ and } y_{j+1} \neq y_1, \mbox{ or} \\
                & \text{there are } j_1, j_2 \geq 1 \text{ such that }\\
                 & y_{j_1}=y_{j_2}, x_{j_1}=x_{j_2}, \text{ and } y_{{j_1}+1} \neq y_{{j_2}+1},  \\
                & \mbox{or } w^{j_1} \in L_1, w^{j_2} \in L_2, \text{ and } y_{{j_1}} = y_{{j_2}}\}.

\end{array}\]

Before constructing an $(L_1,L_2,k)$-refuter that generates short certificates, we first need some notations and observations.
Let $\A=\zug{\Sigma,Q,q_0,\delta,F}$ and $\A'=\langle \Sigma,Q'$, $q_0'$, $\delta',F'\rangle$ be DFAs.
We define the set $F_{\A,\A'}$ of states of $\A$ that are reachable by traversing a word in $L(\A')$. Formally, $q \in F_{\A,\A'}$ iff there is $w \in L(\A')$ such that $\delta^*(q_0,w) = q$, where $\delta^*$ is the extension of $\delta$ to words.
Note that $F_{\A,\A'}$ does not depend on the acceptance condition of $\A$.

\begin{lemma}
\label{faa}
	For every DFAs $\A$ and $\A'$, we have that $L(\A') \subseteq L(\A)$ iff $F_{\A,\A'}\subseteq F$, and $L(\A) \cap L(\A') = \emptyset$ iff $F_{\A,\A'} \subseteq Q \setminus F$.
\end{lemma}

\begin{proof}
	We start with the first claim. If $F_{\A,\A'}\subseteq F$, then for every word $w\in L(\A')$, we have that $\delta^*(q_0,w) \in F$, and so $w \in L(\A)$ and $L(\A')\subseteq L(\A)$. If $F_{\A,\A'}\not \subseteq F$, then there exists a word $w\in L(\A')$ such that $\delta^*(q_0,w) \in Q\setminus F$. Then, $w\in L(\A')\setminus L(\A)$, and so $L(\A')\not \subseteq L(\A)$.

	For the second claim, note that $L(\A)\cap L(\A')=\emptyset$ iff $L(\A') \subseteq \comp(L(\A))$. Let $\tilde{\A}$ be $\A$ with $Q \setminus F$ being the set of accepting states. By the first claim, we have that $L(\A') \subseteq L(\tilde{\A})$ iff $F_{\tilde{\A},\A'}\subseteq Q \setminus F$. Since $\A$ and $\tilde{\A}$ differ only in the acceptance condition, $F_{\tilde{\A},\A'}=F_{\A,\A'}$, and so we are done.
	\hfill \qed
\end{proof}

Lemma~\ref{faa} implies the following characterization of separability by a DFA with a given structure:

\begin{theorem}
	Consider DFAs $\A_1$, $\A_2$, and $\A$. Let $\A=\zug{\Sigma,Q,q_0,\delta,\emptyset}$. For a set $F\subseteq Q$, define $\A_F = \zug{\Sigma,Q,q_0,\delta,F}$. Then, $F_{\A,\A_1} \cap F_{\A,\A_2} = \emptyset$ iff there exists a set $F\subseteq Q$ such that $\A_F$ separates $\zug{\A_1,\A_2}$.
\end{theorem}

\begin{proof}
	By Lemma~\ref{faa}, the DFA $\A_F$ is a separator for $\zug{\A_1,\A_2}$ iff $F_{\A,\A_1} \subseteq F$ and $F_{\A,\A_2}\subseteq Q\setminus F$. If $F_{\A,\A_1}\cap F_{\A,\A_2}=\emptyset$, then $F=F_{\A,\A_1}$ satisfies both containments. In the other direction, if there exists a set $F$ that satisfies both containments, then $F_{\A,\A_1} \cap F_{\A,\A_2} = \emptyset$.
	\hfill \qed
\end{proof}

Consider a DFA $\A=\zug{\Sigma,Q,q_0,\delta,\emptyset}$. If there is no set $F$ such that $\A_F$ is a separator for $\zug{\A_1,\A_2}$, there exists a state $q\in F_{\A,\A_1}\cap F_{\A,\A_2}$. That is, there are words $w_1\in L(\A_1)$ and $w_2\in L(\A_2)$ such that $\delta^*(q_0,w_1)=\delta^*(q_0,w_2)=q$. Note that if Prover follows $\A$, then Refuter can cause the interaction to be a word in $\violatesep(L(\A_1),L(\A_2),k)$ by generating $w_1\cdot \# \cdot w_2 \cdot \#$. Indeed, then the resulting prefix cannot agree with $L(\A_1)$ and $L(\A_2)$. Accordingly, Refuter's strategy is to first force Prover to commit on the transitions of a $k$-DFA, and then to generate $w_1\cdot \# \cdot w_2 \cdot \#$, for the appropriate words $w_1$ and $w_2$. Next, we show how Refuter can force Prover to commit on the transitions of a $k$-DFA.

A legal word $w=x\oplus y$ induces a partial function $\delta_w : [k] \times \Sigma \to [k]$, where for all $j \geq 1$, we have that $y_{j+1}= \delta_w(y_j, x_j)$. Forcing Prover to commit on the transitions of a $k$-DFA amounts to generating a word $w$ for which $\delta_w$ is complete.
\begin{lemma}
\label{expose}
	For every $k \geq 1$, there is a strategy for Refuter that forces Prover to commit on the transitions of a $k$-DFA in $O(k^2\cdot |\Sigma|)$ rounds.
\end{lemma}

\begin{proof}
	Refuter maintains a set $S\subseteq [k]$ of discovered states, and a set $\Delta \subseteq [k]\times \Sigma \times [k]$ of discovered transitions. Note that for every discovered state $q\in S$, Refuter can construct a word $w\in \Sigma^*$ that Prover maps to $q$ using transitions in $\Delta$. Initially, the sets $S$ and $\Delta$ are empty. Prover starts the interaction outputting an initial state $q_0$, and Refuter sets $S=\{q_0\}$.

	Assume that there is an undiscovered transition from one of the discovered states. That is, there exist $q\in S$ and $\sigma \in \Sigma$ such that $\zug{q,\sigma,r}\notin \Delta$ for all $r\in[k]$.  Refuter outputs $w \cdot \sigma\cdot \#$, where $w$ is a word Prover maps to $q$. Then, Prover answers with a state $q'$, and Refuter adds $q'$ to $S$, and $\zug{q,\sigma,q'}$ to $\Delta$.

	Refuter repeats the above process until $\Delta$ is complete. Each of the $k$ states has $|\Sigma|$ outgoing transitions. Refuter exposes one new transition in at most $k + 1$ rounds: A shortest word $w$ that Prover maps to $q$ has length at most $k-1$, then she outputs the letter $\sigma$, and then $\#$. Overall, the number of rounds is at most $k\cdot (k+1)\cdot |\Sigma|$, which is $O(k^2\cdot |\Sigma|)$.
	\hfill \qed
\end{proof}

\begin{lemma}
\label{quadratic}
	Let $\A=\zug{\Sigma,Q,q_0,\delta,F}$ and $\A'=\zug{\Sigma,Q',q_0',\delta',F'}$ be DFAs with $N$ and $N'$ states, respectively. For every state $q\in Q$, if there exists a word $w\in L(\A')$ such that $\delta^*(q_0,w)=q$, then there exists a word $w'\in L(\A')$ such that $\delta^*(q_0,w')=q$ and $|w'|\leq N\cdot N'$.
\end{lemma}

\begin{proof}
	Consider the product DFA ${\cal P}=\A \times \A'$. Let $w\in L(\A')$ be such that $\delta^*(q_0,w)=q$. Let $q'= (\delta')^*(q'_0,w)$. Then, the state $\zug{q,q'}$ of ${\cal P}$ is reachable from $\zug{q_0,q'_0}$. A simple path from $\zug{q_0,q'_0}$ to $\zug{q,q'}$ induces the required word $w'$.
	\hfill \qed
\end{proof}

\begin{theorem}
\label{cert sep}
	Let $L_1,L_2\subseteq \Sigma^*$ be regular languages, and let $N_1=\indx(L_1)$ and $N_2=\indx(L_2)$. For every $k\geq 1$, if $\zug{L_1,L_2}$ is not $k$-DFA-separable, then Refuter can generate a word in $\violatesep(L_1,L_2,k)$ in $O(k^2 \cdot |\Sigma| + k\cdot N_1 + k\cdot N_2)$ rounds.
\end{theorem}

\begin{proof}
	As described in Lemma~\ref{expose}, Refuter can force Prover to commit on a $k$-DFA $\A$ in $O(k^2 \cdot |\Sigma|)$ rounds. Since $\zug{L_1,L_2}$ is not $k$-DFA-separable, there are words $w_1\in L_1,w_2\in L_2$ such that the runs of $\A$ on $w_1$ and on $w_2$ both end in the same state. By Lemma~\ref{quadratic}, there exist such words satisfying $|w_1|\leq k\cdot N_1$ and $|w_2|\leq k\cdot N_2$. Refuter maintains a pair of such words for every $k$-DFA. After the DFA $\A$ is exposed, Refuter outputs the corresponding string $w_1 \cdot \#\cdot w_2 \cdot \#$, which has length at most $k\cdot N_1 + k \cdot N_2 + 2$. Overall, the interaction requires $O(k^2 \cdot |\Sigma| + k\cdot N_1 + k\cdot N_2)$ rounds.
	\hfill \qed
\end{proof}

Recall that when $L_2=\comp(L_1)$, separation coincides with recognizability, with $N_1=N_2=N$. Hence, the $O(N^2)$ lower bound on the length of certificates in Theorem~\ref{interactive refuters bad prefix}, applies also for $(N-1)$-DFA-separation.  Our upper bound for $(N-1)$-DFA-separation in Theorem~\ref{cert sep} includes an extra $|\Sigma|$ factor, as Refuter first forces Prover to commit on all transitions of the claimed DFA. We conjecture that Refuter can do better and force Prover to only to commit on a relevant part of the claimed DFA; namely one in which we can still point to a state $q \in F_{\A,\A_1} \cap F_{\A,\A_2}$ that is reachable via two words $w_1\in L(\A_1)$ and $w_2\in L(\A_2)$. Thus, rather than forcing Prover to commit on all $|\Sigma|$ successors of each state, Refuter forces Prover to commit only on transitions that reveal new states or reveal the required state $q$. Then, the prefix of the certificate that is generated in Lemma~\ref{expose} is only of length $O(N^2)$, making the bound tight. Note that such a lazy exposure of the claimed DFA could be of help also in implementations of algorithms for the DFA identification problem \cite{HV10}.

\section{Discussion and Directions for Future Research}

\paragraph{On the Size of Provers and Refuters.}

Our study of certification focused on the {\em length\/} of certificates.
We did not study the {\em size\/} of the transducers used by Prover and Refuter in order to generate these certificates.
A naive upper bound on the size of such transducers follows from the fact that they are winning strategies in a game played on a deterministic looping automaton for $\violatesep(L,k)$. Such an automaton has to store in its state space the set of transitions committed by Prover, and is thus exponential in $k$. The $(L,k)$-refuter we used for generating short certificates is also exponential in $k$, as it stores in its state space a mapping from the $k+1$ words in $H$ to $[k]$ (see Theorem~\ref{interactive refuters bad prefix}). On the other hand, it is easy to see that Prover can do with a transducer that is polynomial in $k$, as she can follow the transitions of $\RL$.

Interestingly, with a slight change in the setting, we can shift the burden of maintaining the set of transitions committed by Prover from Refuter to Prover. We do this by requiring Prover to reveal new states in her claimed $k$-DFA in an {\em ordered\/} manner: Prover can respond with a state $i \in [k]$ only after she has responded with states $\{1,\ldots,i-1\}$. Formally, we say that $w = x \oplus y \in (\Sigma' \times [k])^* \cup (\Sigma' \times [k])^\omega$, with $x = x_1 \cdot x_2 \cdots$ and $y = y_1 \cdot y_2 \cdots$ is \emph{ordered} iff for all $1 \leq j \leq |w|$ we have $y_{j} \leq \max\{y_l : 1 \leq l < j\} + 1$. Note that if Prover has a winning strategy in a game on $\DFAreal(L,k)$, she also has a winning strategy in a game in which $\DFAreal(L,k)$ is restricted to ordered words. In such a game, however, Refuter can make use of $\RL$ and circumvent the maintenance of subsets of transitions, whereas Prover has to maintain a mapping from the states in $\RL$ to their renaming imposed by the order condition. We leave the analysis of this setting as well as the study of trade-offs between the size of transducers and the length of the certificates to future research.

\paragraph{Infinite words.}

Our setting considers automata on finite words, and it focuses on the number of states required for recognizing a regular language. In \cite{KS21b}, we used a similar methodology for refuting the recognizability of $\omega$-regular languages by automata with limited expressive power.
For example, deterministic {\em B\"uchi\/} automata (DBAs) are less expressive than their non-deterministic counterpart, and a DBA-refuter generates certificates that a given language cannot be recognized by a DBA. Thus, the setting in \cite{KS21b} is of automata on infinite words, and it focuses on expressive power.

Unlike DFAs, which allow polynomial minimization, minimization of DBAs is NP-complete \cite{Sch10}. Combining our setting here with the one in \cite{KS21b} would enable the certification and refutation of {\em $k$-DBA-recognizability}, namely recognizability by a DBA with $k$ states. The NP-hardness of DBA minimization makes this combination very interesting. In particular, there are interesting connections between polynomial certificates and possible membership of DBA minimization in co-NP, as well as connections between size of certificates and succinctness of the different classes of automata.

\bibliography{../ok}
\bibliographystyle{plain}

\end{document}